\documentclass[10pt]{article}
\usepackage{amsthm,amsmath,amssymb,upgreek}
\usepackage{subfig}
\usepackage{paralist}
\usepackage[usenames,dvipsnames]{color}
\usepackage[pdftex,breaklinks,colorlinks,
    citecolor={BlueViolet}, linkcolor={Blue},urlcolor=Maroon]{hyperref}  
\usepackage{tikz,pgfkeys}
\usepackage{tkz-graph}
\usetikzlibrary{arrows,shapes}
\usepackage{graphicx}
\usepackage{charter,eulervm}%
\usepackage[final,expansion=alltext,protrusion=true]{microtype}
\graphicspath{{./figs/}}

\theoremstyle{plain} 
\newtheorem{theorem}{Theorem}[section]
\newtheorem{lemma}[theorem]{Lemma}
\newtheorem{corollary}[theorem]{Corollary}
\newtheorem{proposition}[theorem]{Proposition}
\newtheorem{reduction}{Rule}

\newcommand{\uivd}{{unit interval vertex deletion}}
\newcommand{\fis}{{forbidden induced subgraph}}
\newcommand{\lp}[1]{\ensuremath{{\mathtt{lp}(#1)}}}
\newcommand{\rp}[1]{\ensuremath{{\mathtt{rp}(#1)}}}
\newcommand{\stpath}[2]{$#1$-$#2$ path}
\newcommand{\stsep}[2]{$#1$-$#2$ separator}

\title{Unit Interval Vertex Deletion: Fewer Vertices are
  Relevant\thanks{Supported in part by the Hong Kong Research Grants
    Council (RGC) under grant PolyU 252026/15E and the National
    Natural Science Foundation of China (NSFC) under grants 61572414
    and 61420106009.}}

\author{Yuping Ke\thanks{School of Information Science and
    Engineering, Central South University, Changsha, China.  
    \href{jxwang@mail.csu.edu.cn} {\tt jxwang@mail.csu.edu.cn}.
  } \thanks{Department of Computing, Hong
    Kong Polytechnic University, Hong Kong,
    China. \href{mailto:yixin.cao@polyu.edu.hk} {\tt
      yixin.cao@polyu.edu.hk}.}
  \and
  \addtocounter{footnote}{-1} Yixin Cao\footnotemark 
  \and
  \addtocounter{footnote}{-1} Xiating Ouyang\footnotemark 
  \and
  \addtocounter{footnote}{-2} Jianxin Wang\footnotemark 
 }

\begin{document}
\maketitle
\begin{abstract}
  The unit interval vertex deletion problem asks for a set of at most $k$ vertices whose deletion from an $n$-vertex graph makes it a unit interval graph.  We develop an $O(k^4)$-vertex kernel for the problem, significantly improving the $O(k^{53})$-vertex kernel of Fomin, Saurabh, and Villanger [ESA'12; SIAM J.\ Discrete Math 27(2013)]. We introduce a novel way of organizing cliques of a unit interval graph.  Our constructive proof for the correctness of our algorithm, using interval models, greatly simplifies the destructive proofs, based on forbidden induced subgraphs, for similar problems in literature.
\end{abstract}

\section{Introduction}\label{sec:intro}

A graph is a \emph{unit interval graph} if its vertices can be
assigned to unit-length intervals on the real line such that there is
an edge between two vertices if and only if their corresponding
intervals intersect.  Given a graph $G$ and an integer $k$, the
\emph{unit interval vertex deletion} problem asks whether there is a
set of at most $k$ vertices whose deletion makes $G$ a unit interval
graph.  According to Lewis and
Yannakakis~\cite{lewis-80-node-deletion-np}, this problem is
NP-complete.

This paper approaches this problem by studying its kernelization.
Given an instance $(G, k)$ of \uivd{}, a {\em kernelization algorithm}
produces in polynomial time an ``equivalent'' instance $(G', k')$ such
that $k' \leq k$ and the {\em kernel size} (i.e., the number of
vertices in $G'$) is upper bounded by some function of $k'$.  Fomin et
al.~\cite{fomin-13-kernel-pivd} presented an $O(k^{53})$-vertex kernel
for the problem, which we improve to the following, where $n$ and $m$
denote respectively the numbers of vertices and edges of the input
graph.
\begin{theorem}\label{thm:main}
  The \uivd{} problem has an $O(k^4)$-vertex kernel, and it can be
  produced in $O(n m + n^2)$ time.
\end{theorem}

\tikzstyle{vertex}  = [{fill=cyan,circle,cyan,draw,inner sep=1pt}]
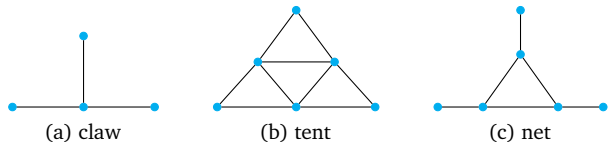
\begin{figure*}[h]
  \centering\footnotesize
  \subfloat[claw]{\label{fig:claw}
    \begin{tikzpicture}[scale=.27]
    \node [vertex] (a1) at (-3.5, 0) {};
    \node [vertex] (v) at (0, 0) {};
    \node [vertex] (b1) at (3.5, 0) {};
    \node [vertex] (c) at (0,3.5) {};
    \draw[] (a1) -- (v) -- (b1);
    \draw[] (v) -- (c);
    \end{tikzpicture}
  }
  \qquad
  \subfloat[tent]{\label{fig:tent}
    \begin{tikzpicture}[scale=.2]
    \node [vertex] (s) at (0,6.44) {};
    \node [vertex] (a) at (-5.25,0) {};
    \node [vertex] (a1) at (0, 0) {};
    \node [vertex] (b) at (5.25,0) {};
    \node [vertex] (c1) at (-2.55,3) {};
    \node [vertex] (c2) at (2.55,3) {};
    \draw[] (a) -- (a1) -- (b) -- (c2) -- (s) -- (c1) -- (a);
    \draw[] (c1) -- (c2) -- (a1) -- (c1);
    \end{tikzpicture}
  }
  \qquad
  \subfloat[net]{\label{fig:net}
    \begin{tikzpicture}[scale=.2]
    \node [vertex] (s) at (0,6.44) {};
    \node [vertex] (a) at (-5.5, 0) {};
    \node [vertex] (a1) at (-2.5, 0) {};
    \node [vertex] (b1) at (2.5,0) {};
    \node [vertex] (b) at (5.5,0) {};
    \node [vertex] (c) at (0,3.5) {};
    \draw[] (a) -- (a1) -- (b1) -- (b);
    \draw[] (c) -- (s);
    \draw[] (a1) -- (c) -- (b1);
    \end{tikzpicture}
  }
  \caption{Forbidden induced graphs.}
  \label{fig:fis}
\end{figure*}
The structures of unit interval graphs have been well studied and well
understood.  It is known that a graph is a unit interval graph if and
only if it contains no claw, net, tent, (as depicted in
Figure~\ref{fig:fis},) or any hole (i.e., an induced cycle on at least
four vertices) \cite{roberts-69-indifference-graphs,
  wegner-67-dissertation}.  One can decide in linear time whether a
given graph is a unit interval graph; if it is not, we can obtain a
forbidden induced subgraph in the same time
\cite{hell-04-certifying-proper-interval}.  The \uivd{} problem can
then be equivalently defined as finding a set of at most $k$ vertices
that hits all \fis{s} of the input graph.

The vertex deletion problem has been defined on many other graph
classes and has been intensively studied.  It is the easy case when
the objective graph class has a finite set of \fis{s}: The sunflower
lemma implies a polynomial kernel for the vertex deletion problem to
this graph class \cite{flum-grohe-06}.  However, kernels produced by
the sunflower lemma tend to be very large.  Furthermore, most
interesting graph classes have an infinite number of \fis{s}, hence
the hard case.  Another approach that works for both (and even edge
modification problems) is to start from a modulator, i.e., a set of vertices
whose deletion leaves a graph in the objective graph class.
With a modulator, we are allowed to use the properties of unit
interval graphs to analyze the rest of the graph.  What is important
is the interaction between other vertices with the modulator, and thus
through its analysis we can identify irrelevant vertices, thereby
producing the kernel \cite{drange-14-kernel-trivially-perfect,
  jansen-16-approximation-and-kernelization-chordal-deletion}.

Since holes of any length are forbidden in unit interval graphs, our
problem is clearly the hard case.  Hence both Fomin et
al.~\cite{fomin-13-kernel-pivd} and this paper use the modulator
approach.  For both of us, the modulator consists of two parts, first
a set of vertices that hits all \emph{small} \fis{s},---we use
different thresholds for bing small,---and the second an optimal
hitting set for long holes in the remaining graph.  Recall that long
holes behave very nicely in a graph free of small \fis{s}; for
example, a minimum hitting set for them can be found in linear time
\cite{cao-15-unit-interval-editing}.  Thus, our main concern is the
first part.  What differentiates these two algorithms is how they are
carried out.  Fomin et al.~\cite{fomin-13-kernel-pivd} used the
sunflower lemma to produce the modulator, while we use a
constant-approximation algorithm.

We only need to proceed when the approximation algorithm produces a
solution of $O(k)$ vertices.  Thus, our modulator has a linear size,
which is in a sharp contrast with the huge modulator of Fomin et
al.~\cite{fomin-13-kernel-pivd} produced by the sunflower lemma.  On
the other hand, their modulator has an extra property that is not
shared by ours.  It guarantees that one only needs to care about small
forbidden induced subgraphs \emph{inside the modulator}, thereby
saving them a lot of trouble in the selection of relevant vertices.
Our modulator nevertheless does not have this property.  Therefore,
the interaction of the modulator with the rest of the graph is far
more complicated in our case, and the analysis is fundamentally
different.  In particular, the main technical difficulties present
themselves exactly at the analysis of the small \fis{s}.

This is exactly where our main technical ideas appear, which result in
an algorithm and analysis significantly simpler than that of Fomin et
al.~\cite{fomin-13-kernel-pivd}.  Let ($G, k$) be the input instance
and let $M$ be the modulator.  Our first idea is to \emph{partition}
the vertices of the unit interval subgraph $G - M$ into cliques and
organize them in a linear way such that vertices in each clique are
adjacent to only vertices in its two neighboring cliques.  This is
quite different from the widely used clique path decomposition,
because in a clique path decomposition, (1) a vertex can appear in
more than one cliques; and (2) two vertices in cliques that are far
away can be adjacent.  Our reduction rules ensure that if ($G, k$) is
reduced, then there cannot be more than $O(k^2)$ cliques in the
partition of $G - M$.  From each of these cliques at most $O(k^3)$
vertices are relevant.  This implies a kernel of $O(k^5)$ vertices,
and a more careful analysis leads to the smaller size claimed in
Theorem~\ref{thm:main}.

Our second idea appears in the proof of the correctness of our
algorithm.  Our reduction rules are rather elementary and
self-explanatory.  The main step is to prove the irrelevance of the
other vertices.  We may assume that ($G, k$) is already reduced, and
let $G'$ denote the subgraph induced by the relevant vertices and the
modulator.  We need to show that if there is a solution $V_-$ of size
at most $k$ to our kernel $G'$, then it is a solution to $G$.  Instead
of showing the nonexistence of forbidden induced subgraph in $G -
V_-$, (which would necessarily lead to an endless list of cases,) we
build a unit interval model for $G - V_-$ out of a unit interval model
for $G' - V_-$.

In a companion paper \cite{cao-16-kernel-interval-vertex-deletion}, we
have also developed a polynomial kernel for the interval vertex
deletion problem, which is arguably more challenging and was only
recently shown to be fixed-parameter tractable
\cite{cao-15-interval-deletion}.  It extends the ideas from this
paper.  It also starts from a modulator produced by the
constant-approximation algorithm
\cite{cao-16-almost-interval-recognition}, but the analysis is far
more complicated.  In particular, the novel clique partition, which
plays a crucial role in simplifying the analysis, does not apply there
in any way we know of.
Jansen and
Pilipczuk~\cite{jansen-16-approximation-and-kernelization-chordal-deletion}
recently studied the kernelization of the chordal vertex deletion
problem and produced the first polynomial kernel.  They also used an
approximation solution as the modulator, for which they had to first
design a polynomial-time approximation algorithm.  Their kernel has
also a huge size, $O(k^{162})$ vertices.

Let us also mention the related parameterized algorithms (i.e.,
algorithms running in time $O(f(k)\cdot n^{O(1)})$ for some function
$f$ independent of $n$) for the problem, which have undergone a
sequence of improvements.  Recall that chordal graphs are those graphs
containing no holes, and thus unit interval graphs are a subclass of
chordal graphs.  Toward a parameterized algorithm with $f(k) =
\Omega(6^k)$, one can always dispose of all the claws, nets, and tents
from the input graph, and then call the algorithm for the chordal
vertex deletion problem \cite{marx-10-chordal-deletion,
  cao-16-chordal-editing} to break all holes in the remaining graph.
Direct algorithms for \uivd\ were later reported in
\cite{bevern-10-pivd, villanger-13-pivd,
  cao-15-unit-interval-editing}, and the current best algorithm runs
in time $O(6^k\cdot (n + m))$.  All the three direct algorithms use a
two-phase approach.  In \cite{cao-15-unit-interval-editing}, for
example, the first phase breaks all claws, nets, tents, and $C_4$'s,
while the second phase deals with the remaining \{claw, net, tent,
$C_4$\}-free graphs, on which the problem can be solved in polynomial
time.  A simple adaptation of this approach leads to an $O(n m +
n^2)$-time $6$-approximation algorithm.

\paragraph{Organization.} The rest of the paper is organized as
follows.  Section~\ref{sec:clique-partition} introduces the clique
partition and its properties.  Sections~\ref{sec:reductions} presents
three simple reduction rules.  Section~\ref{sec:selection} finishes
the kernel by handpicking vertices from the reduced graph.
Section~\ref{sec:remarks} closes this paper by discussing
implementation issues.

\section{The clique partition}\label{sec:clique-partition}

All graphs discussed in this paper are undirected and simple.  A graph
$G$ is given by its vertex set $V(G)$ and edge set $E(G)$, whose
cardinalities will be denoted by $n$ and $m$ respectively.  All input
graphs in this paper are assumed to be connected, hence $n = O(m)$
whenever $n > 1$.  For $\ell\ge 4$, we use $C_\ell$ to denote a hole
on $\ell$ vertices.  Chordal graphs are precisely $\{C_\ell:\ell\ge
4\}$-free graphs.

In this paper, all intervals are closed.  An \emph{interval graph} is
the intersection graph of a set of intervals on the real line.  The
set of intervals, called an {\em interval model}, can be specified by
their $2 n$ endpoints: The interval $I(v)$ for vertex $v$ is given by
$[\lp{v}, \rp{v}]$, where \lp{v} and \rp{v} are its \emph{left and
  right endpoints} respectively.  It always holds $\lp{v} < \rp{v}$.
No distinct intervals are allowed to share an endpoint in the same
model; note that this restriction does not sacrifice any generality.
A graph is a {\em unit interval graph} if it has a {\em unit interval
  model}, where every interval has length one.  An interval model is
\emph{proper} if no interval in it properly contains another interval.
It is easy to see that every unit interval model is proper; a
nontrivial observation of
Roberts~\cite{roberts-69-indifference-graphs} is that every graph
having a proper interval model also has a unit interval model.

\begin{figure*}[h]
  \centering\footnotesize
  \includegraphics{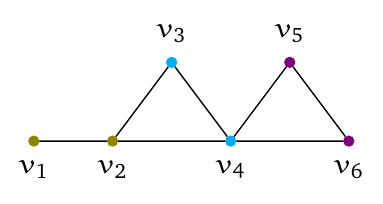}   
  \qquad
  \includegraphics{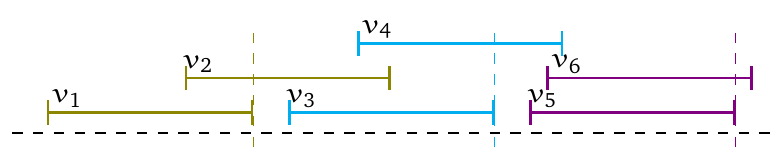}   
  \caption{A unit interval graph and its unit interval model.  The
    proper interval ordering decided by this model is $\langle v_1,
    v_2, v_3, v_4, v_5, v_6 \rangle$, which partitions the graph into
    cliques $\{v_1, v_2\}$, $\{v_3, v_4\}$, and $\{v_5, v_6\}$,
    corresponding to the three dashed vertical lines respectively.}
  \label{fig:clique-partition}
\end{figure*}

Note that in a proper interval model, if $\lp{u} <\lp{v}$, then
$\rp{u} <\rp{v}$ as well.  Therefore, it makes sense to talk about the
left-right relationship of the intervals.  
If we read the vertices by the ordering of their intervals, we end
with a \emph{proper interval ordering} of the graph
\cite{looges-93-greedy-algorithms-uig}.
The following property is an easy consequence of
the definition of proper interval models and proper interval ordering.
\begin{proposition}\label{lem:proper-interval-ordering}
  Let $v_1, \ldots, v_n$ be a proper interval ordering of a unit
  interval graph $G$.  For every $1\le i < j \le n$, if $v_i v_j\in
  E(G)$, then $\{v_i, \ldots, v_j\}$ is a clique.
\end{proposition}

Fixing a unit interval model $\cal I$ for a unit interval graph $G$,
we can greedily partition its vertices into a sequence of cliques.
Initially all vertices are unassigned.  We repetitively choose the
unassigned vertex $v$ with the leftmost interval, and take all
vertices whose intervals containing $\rp{v}$; this set is clearly a
clique.  We proceed until the graph becomes empty.  Let ${\cal K} =
\{K_1, \ldots, K_t\}$ be the set of cliques obtained in the order.
See Figure~\ref{fig:clique-partition} for an example.  One should be
noted that the cliques are not maximal in general; in particular the
last vertex of $K_{i - 1}$ might be adjacent to all vertices in
$K_{i}$, e.g., both the second the third cliques in
Figure~\ref{fig:clique-partition}.  The following proposition and its
corollary characterize this partition, and facilitate our analysis of
the kernel size.
\begin{proposition}\label{lem:partition}
  Let $K_1, \ldots, K_t$ be the clique partition of a unit interval
  model $G$.  For each $1< i < t$, it holds that $N(K_i)\subset
  K_{i-1}\cup K_{i+1}$.  Moreover, $N(K_1)\subseteq K_{2}$,
  $N(K_t)\subset K_{t - 1}$.
\end{proposition}
\begin{proof}
  Let $v\in K_i$.  By the definition of clique partition, $v$ is
  nonadjacent to the first vertex of $K_{i - 1}$.  By
  Proposition~\ref{lem:proper-interval-ordering}, no neighbor of $v$
  comes before the first vertex of $K_{i - 1}$.  Thus, if a neighbor
  of $v$ is before $K_i$, it has to be in $K_{i - 1}$.  On the other
  hand, any neighbor of $v$ after $K_i$ is either the first vertex of
  $K_{i +1}$ or adjacent to it, hence in $K_{i + 1}$.  The two border
  cases follow similarly.
\end{proof}

\begin{corollary}\label{lem:distance}
  Let $K_1, \ldots, K_t$ be the clique partition of a unit interval
  model $G$.  For each pair of vertices $u \in K_i$ and $v\in K_j$
  with $1\le i \le j \le t$, the distance between $u$ and $v$ is at
  least $j - i$.
\end{corollary}

The following fact will be used in the correctness proof of our main
reduction rule.  Here by \emph{contracting} a clique $K_i$ ($1 < i <
t$), we mean the operations of deleting all vertices in $K_i$, and
adding all possible edges to connect its neighbors in $K_{i - 1}$ and
in $K_{i + 1}$.

\begin{figure*}[h]
  \centering\footnotesize
    \includegraphics{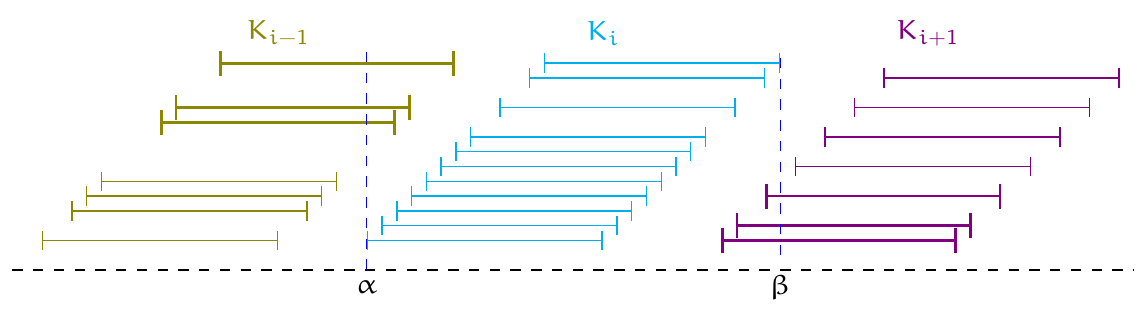}   

    (a) The original unit interval model for $G$.\\[3mm]

    \includegraphics{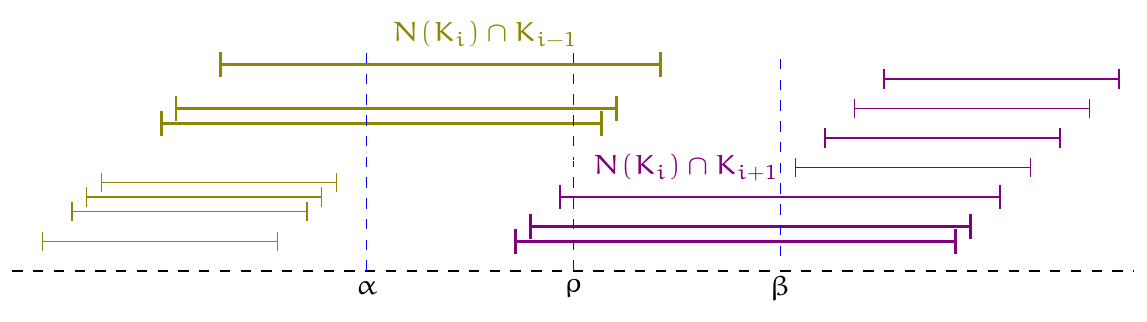}   

    (b) The new proper interval model for the graph obtained by
    contracting $K_i$.  \\Only intervals for vertices in $N(K_i)\cap
    K_{i - 1}$ and $N(K_i)\cap K_{i + 1}$, which are thick, are
    extended.
    \caption{Illustration for Proposition~\ref{lem:contraction}.}
  \label{fig:contraction}
\end{figure*}
\begin{proposition}\label{lem:contraction}
  Let $K_1, \ldots, K_t$ be the clique partition of a unit interval
  graph $G$.  For each $1 < i < t$, the graph obtained by contracting
  $K_i$ is still a unit interval graph.
\end{proposition}
\begin{proof}
  Let $G'$ be the graph obtained by contracting $K_i$; then $V(G') =
  V(G) \setminus K_i$ and $E(G') = E(G - K_i) \cup \big( (N(K_i)\cap
  K_{i - 1}) \times (N(K_i)\cap K_{i + 1}) \big)$.  We build a proper
  interval model for $G'$ as follows.  Let $\alpha$ be the left
  endpoint of the first vertex in $K_i$, and let $\beta$ be the right
  endpoint of the last vertex in $K_i$; note that $\bigcup_{v\in K_i}
  I(v) = [\alpha, \beta]$.  We choose an arbitrary point $\rho$
  between $\alpha$ and $\beta$.  For vertices in $N(K_i)\cap K_{i -
    1}$, we change their right endpoints to between $\rho$ and $\beta$
  while keeping their orders.  Likewise, for vertices in $N(K_i)\cap
  K_{i + 1}$, we change their left endpoints to between $\alpha$ and
  $\rho$ while keeping their orders.  See
  Figure~\ref{fig:contraction}.  Note that only intervals for
  $N(K_i)\cap K_{i - 1}$ and $N(K_i)\cap K_{i + 1}$ are extended, and
  all the extensions are made in $[\alpha, \beta]$, where is disjoint
  from all other intervals.  It is then easy to verify that the new
  interval model is proper and represents $G'$.
\end{proof}

In passing we should point out that neither the ordering nor the
clique partition is unique in general.\footnote{Even so, one can say that it
  is \emph{almost unique}, in the sense that there can be at most two
  partitions: true twins (vertices with the same closed neighborhood)
  always reside in the same clique, while on a true-twin-free graph,
  there is only one ordering up to full reversal
  \cite{deng-96-proper-interval-and-cag, hsu-95-recognition-cag}.}

\section{The reduction rules}\label{sec:reductions}

Let ($G, k$) be an instance of the \uivd{} problem.  We start by
calling the $6$-approximation
algorithm~\cite{cao-15-unit-interval-editing} to find an approximation
solution $M$ to $G$.  If $|M| > 6 k$, then we return a trivial
no-instance.  We may assume henceforth $|M| \le 6 k$, and fix a unit
interval model for $G - M$.  Let $\sigma = \langle v_1, v_2, \ldots,
v_{n - |M|} \rangle$ be the proper interval ordering and ${\cal K} =
\{K_1, \ldots, K_t\}$ the clique partition derived from this model.

As an easy consequence of Proposition~\ref{lem:partition}, any vertex
in $M$ that is adjacent to five or more cliques in $\cal K$ is the
center of some claw.  This observation inspires the following two
reduction rules, whose correctness is straightforward: If you do not
delete the vertex $v$ itself, then you have to delete at least $k + 1$
vertices to break all claws involving $v$.

\begin{reduction}\label{rule:star-1}
  If there exists a vertex $v\in M$  that is adjacent to at least
   $k + 5$ cliques in $\cal K$, then delete $v$ and decrease $k$ by $1$.
\end{reduction}

\begin{reduction}\label{rule:star-2}
  If there exist a vertex $v\in M$ and at least five cliques in $\cal
  K$ such that each of these cliques contains at least $k + 1$ neighbors of
  $v$, then delete $v$ and decrease $k$ by $1$.
\end{reduction}

Note that the diameter of a claw, net, or tent is at most three.  The
following is immediate from Corollary~\ref{lem:distance} and the fact
that any claw, net, or tent of $G$ needs to intersect $M$.
\begin{corollary}
  If there is $3\le \ell\le t - 2$ such that $M$ is nonadjacent to
  $K_i$ for $\ell - 2 \le i \le \ell + 2$, then no vertex in $K_\ell$
  is contained in a claw, net, or tent.
\end{corollary}

Therefore, if there exists a long sequence of cliques that are
nonadjacent to $M$, then most vertices in the middle can only
participate in holes.  Such a hole, if it exists, necessarily visits
all these cliques, and in particular, it must enter from one end
clique and leave at the other end.  Moreover, it visits every clique
in between, and contains one or two vertices from each of them. If we
delete vertices from these cliques (for the purpose of breaking these
holes), then we would choose a minimum separator.  This observation
motivates the next reduction rule.  It has been observed in a more
general form in \cite[Reduction 2, Section 6]{cao-16-chordal-editing};
with the clique partition, we can simplify it to the following form.
Recall that any minimal separator of a unit interval graph is a
clique, which cannot intersect more than two cliques in $\cal K$.

\begin{reduction}\label{rule:smooth}
  Let $K_{i - 3},\ldots, K_{i + 3}$ be 7 consecutive cliques in $\cal
  K$ that are nonadjacent to $M$.  Let $u$ be the last vertex in $K_{i
    - 2}$ and let $v$ be the first vertex in $K_{i + 2}$.  We take a
  minimum \stsep{u}{v} $S$ in $G - M$, and let $\ell \in \{i - 1, i, i
  + 1\}$ be that $K_\ell$ is disjoint from $S$.  Contract $K_\ell$.
\end{reduction}

\begin{lemma}
  Rule~\ref{rule:smooth} is safe: ($G, k$) is a yes-instance if
  and only if ($G', k$) is a yes-instance, where $G'$ is the resulting
  graph.
\end{lemma}
\begin{proof}
  By {Proposition}~\ref{lem:contraction}, $G' - M$ is still a unit
  interval graph.  Thus, every forbidden induced subgraph in $G'$
  needs to intersect $M$.  Since we have only added edges between
  $K_{\ell - 1}$ and $K_{\ell + 1}$, for every vertex in them, its
  distance to $M$ is at least three.  There cannot be any claw, net,
  or tent in $G'$ involving vertices from both $K_{\ell - 1}$ and
  $K_{\ell + 1}$.  Therefore, we only need to take care of holes in
  the proof.

  Suppose to the contrary of the only if direction that ($G, k$) is a
  yes-instance but ($G', k$) is not.  Let $V_-$ be a solution to ($G,
  k$).  Then there is necessarily a hole of $G'$ that visits at least
  one edge added by the reduction; let it be $x y$ with $x\in K_{\ell
    - 1}$ and $y\in K_{\ell + 1}$.  Since $x,y$ are nonadjacent to
  $M$, their other neighbors on the hole must both belong to
  $V(G)\setminus M$ as well; denote them by $x'$ and $y'$
  respectively.  The ordering of these four vertices has to be $x'
  <_\sigma x <_\sigma y <_\sigma y'$.  There must be an \stpath{x}{y}
  in $G - M$ using only vertices in $K_{\ell - 1}, K_{\ell}, K_{\ell +
    1}$.  Its inner vertices are not adjacent to any vertex in this
  hole, except $x, y$ themselves.  Thus, we end with a hole of $G -
  V_-$, a contradiction.

  On the other hand, if $V_-$ is a solution to $G'$ but there is hole
  in $G - V_-$, then there must be a hole visiting a vertex deleted by
  the reduction.  This hole necessarily visit $N(K_\ell)\cap K_{\ell -
    1}$ and $N(K_\ell)\cap K_{\ell + 1}$.  But then after the
  reduction, its remaining vertices form a hole of $G' - V_-$: Note
  that the original hole has to visit $M$ and hence has length larger
  than $4$.
\end{proof}

\begin{lemma}
  Each of the three reduction rules can be applied in $O(m )$ time.
\end{lemma}
\begin{proof}
  We can mark first the vertices in $M$, and then go through the
  adjacency list of each vertex in $G - M$ in the proper interval
  order.  During this process we can record (1) for each vertex $v\in
  M$, how many cliques in $\cal K$ are adjacent to $v$, and how many
  of them contain $k + 1$ or more neighbors of $v$; and (2) for each
  clique $K\in \cal K$, whether it is adjacent to $M$.  The process
  checks the adjacency list of each vertex once, and thus it takes
  $O(m)$ time in total.  With this information, we can decide which of
  the three reduction rules is applicable, and if yes, apply it in the
  same time.
\end{proof}

After the application of the reduction rules, it is possible that the
rest of $M$ is no longer a $6$-approximation of the reduced graph.
Therefore, we need to re-calculate the modulator.  This would
nevertheless take $O(n^2 m)$ time.  We defer the detailed for an
efficient implementation to Section~\ref{sec:remarks}.

\section{The kernel}\label{sec:selection}

Let ($G, k$) be a \emph{reduced instance} with respect to modulator
$M$, i.e., none of Rules~\ref{rule:star-1}--\ref{rule:smooth} can be
applied to $G$.  Recall that $\sigma$ is the fixed proper interval
ordering of $G - M$, and ${\cal K} = K_1, \ldots, K_t$ is the clique
partition of $G - M$.
We now pick up vertices from $G$ to make the kernel.  The idea is to
pick as few as possible vertices that are relevant, i.e., from
\emph{each type of vertices} (to be defined later) we choose $k + 1$,
which ensures that if any vertex from this type is not picked, then at
least one picked vertex is not deleted by a solution of size at most
$k$.  Note that it is possible that there are less than $k + 1$
vertices in some type, and then we pick all of them.
See Figure~\ref{fig:k0} for an example.

\begin{figure*}[h]
  \centering\footnotesize
    \includegraphics{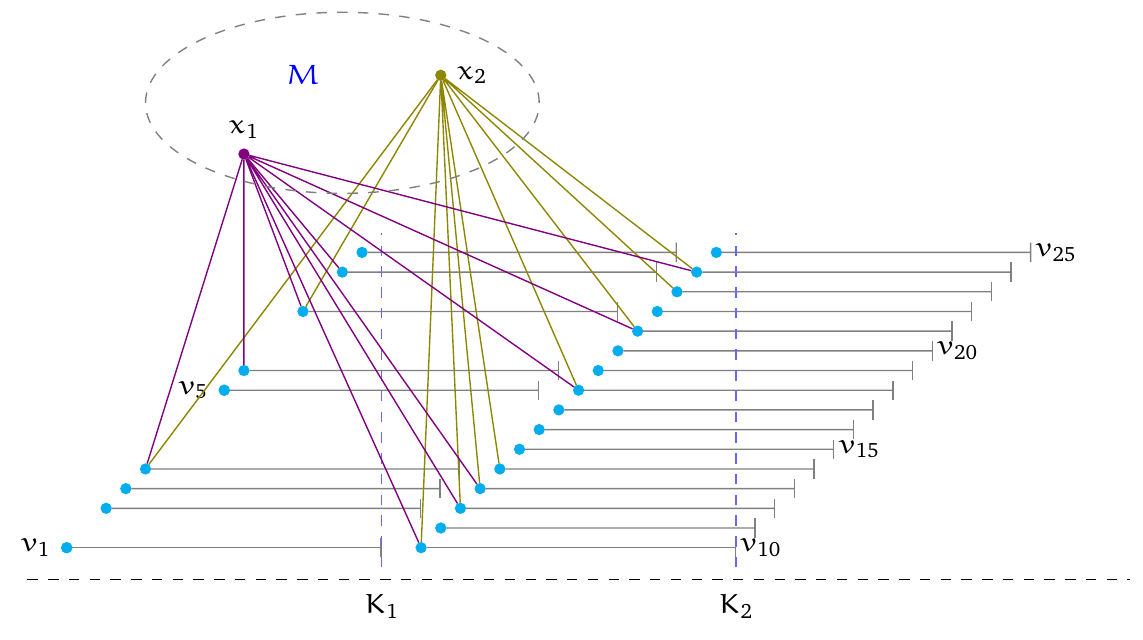}   
    \caption{Illustration for picking vertices.  The modulator $M$
      consists of $x_1$ and $x_2$, and other $25$ vertices are in
      $V(G)\setminus M$.  Edges in $G - M$ are not drawn: The
      intervals, whose left endpoints coincide the vertices they
      represent, make a unit interval model for $G - M$.  The vertices
      in $V(G)\setminus M$ are thus partitioned into two cliques,
      namely, $K_1 = \{v_1, \ldots, v_{9}\}$ and $K_2 = \{v_{10},
      \ldots, v_{25}\}$.  For $k = 2$, we have
      \\[2mm]
      $K_2^1({x_1, x_2}) = \{v_{10}, v_{12}, v_{13}\}\cup \{v_{18},
      v_{21}, v_{24}\}$; $K_2^1({x_1, \overline x_2}) = \emptyset$;
      $K_2^1({ \overline x_1, x_2}) = \{v_{14}, v_{23}\}$;
      $K_2^1({\overline x_1, \overline x_2}) = \{v_{11}, v_{15},
      v_{16}\}\cup \{v_{20}, v_{22}, v_{25}\}$.
      \\[2mm]
      $K_2^2(x_2, v_{6}) = \{v_{12}, v_{13}, v_{14}\}$; $K_2^2(x_2,
      v_{8}) = \{v_{14}, v_{18}, v_{21}\}$; $K_2^2(x_2, v_{9}) =
      \{v_{18}, v_{21}, v_{23}\}$; $K_1^2(x_2, v_{11}) = \{v_{3},
      v_{7}\}$; $K_1^2(x_2, v_{15}) = K_1^2(x_2, v_{16}) = \{v_{7}\}$.
      \\[2mm]
      {$K^3_2(x_1, v_4) = \{v_{13}, v_{18}, v_{21}\}$; $K^3_2(x_1,
        v_6) = \{v_{18}, v_{21}, v_{24}\}$; $K^3_2(x_1, v_7) =
        \{v_{21}, v_{24}\}$; $K^3_1(x_1, v_{24}) = \{v_{6}, v_{7},
        v_{8}\}$; $K^3_1(x_1, v_{21}) = \{v_{4}, v_{6}, v_{7}\}$;
        $K^3_1(x_1, v_{18}) = \{v_{4}, v_{6}\}$.}
      \\[2mm]
      $K^4_2(x_1, v_6) = \{v_{15}, v_{16}, v_{17}\}$; $K^4_2(x_1, v_7)
      = \{v_{17}, v_{19}, v_{20}\}$; $K^4_2(x_1, v_8) = \{v_{19},
      v_{20}, v_{22}\}$. }
  \label{fig:k0}
\end{figure*}

First, for each pair of vertices $x_1, x_2$ in $M$ and each $i = 1,
\ldots, t$, we consider the (non)neighbors of $x_1, x_2$ in $K_i$.  We
pick the first and last $k + 1$ vertices from $K_i$ for each of the
four patterns---adjacent to both; adjacent to only $x_1$; adjacent to
only $x_2$; and adjacent to neither.  Let them be denoted by
$K_i^1({x_1, x_2})$, $K_i^1({x_1, \overline x_2})$, $K_i^1({ \overline
  x_1, x_2})$, and $K_i^1({\overline x_1, \overline x_2})$
respectively.
{Also, let $K_i^1({x})$ denote the first $k + 1$ and the last $k + 1$
  of $\bigcup_{y\in M\setminus \{x\}} \big( K_i^1({x, y}) \cup
  K_i^1({x, \overline y}) \big)$.}

Second, For each $x\in M$, each $i = 2, \ldots, t$, and each of the
last $k + 1$ non-neighbors $y$ of $x$ in $K_{i - 1}$, we pick the last $k
+ 1$ common neighbors of $x$ and $y$ in $K_i$; for each $x\in M$, each
$i = 1, \ldots, t - 1$, and each of the first $k + 1$ non-neighbors
$y$ of $x$ in $K_{i + 1}$, we pick the first $k + 1$ common neighbors of
$x$ and $y$ in $K_i$.  Let them be denoted by $K_i^2(x, y)$.

Third, for each $x\in M$, each $i = 2, \ldots, t$, and each of the
first $k + 1$ neighbors $y$ of $x$ in $K_{i - 1}$, we pick the first
$k + 1$ vertices in $K_i$ that are neighbors of $x$ but not $y$; for
each $x\in M$, each $i = 1, \ldots, t - 1$, and each of the last $k +
1$ neighbors $y$ of $x$ in $K_{i + 1}$, we pick the last $k + 1$
vertices in $K_i$ that are neighbors of $x$ but not $y$.
Let them be denoted by $K_i^3(x, y)$.

Fourth, for each $x\in M$, each $i = 2, \ldots, t$, and each of the
last $k + 1$ neighbors $y$ of $x$ in $K_{i - 1}$, we pick the last $k +
1$ vertices in $K_i$ that are neighbors of $y$ but not $x$; for
each $x\in M$, each $i = 1, \ldots, t - 1$, and each of the first $k +
1$ neighbors $y$ of $x$ in $K_{i + 1}$, we pick the first $k + 1$
vertices in $K_i$ that are neighbors of $y$ but not $x$.
Let them be denoted by $K_i^4(x, y)$.

Finally, for each three pairwise nonadjacent vertices in $M$, we
arbitrarily pick $k + 1$ common neighbors of them in $V(G)\setminus
M$; and for each triple of vertices in $M$ that induces a $P_3$, we
arbitrarily pick $k + 1$ vertices in $V(G)\setminus M$ that are
adjacent to only the center vertex among them.  Let them be denoted by
$V_0$.

Let $K$ be a clique in $\cal K$.  If $|K| \le 2 k + 2$, then all its
vertices have been picked.  We consider then the nontrivial case,
i.e., when $|K| > 2 k + 2$.  The first and last $k + 1$ vertices of
$K$ are always picked; hence at least $2 k + 2$ vertices are picked
from $K$.  Likewise, the first and the last $k + 1$ vertices in $K$
that are nonadjacent to $M$ are always picked; so are the first and
the last $k + 1$ neighbors in $K$ for each $x\in M$.  Moreover, if a
vertex $v$ satisfies the conditions of any particular set but is not
picked, then we have picked from the set $2(k + 1)$ vertices, of which
$k + 1$ are to the left of $v$, and $k + 1$ are to the right of $v$.

Let $G'$ be induced by the picked vertices together with $M$.  We now
calculate the cardinality of $V(G')$.  There are $O{|M| \choose 2} *
O(k) + O(|M|) * O(k) * O(k) = O(k^3)$ vertices picked from each
clique.  On the other hand, the number $t$ of cliques in $\cal K$ is
$O(k^2)$, as otherwise one of Rules~\ref{rule:star-1}
and~\ref{rule:smooth} must be applicable.  Together with at most ${|M|
  \choose 3} * (k + 1) * 2 = O(k^4)$ vertices in $V_0$, and $O(k)$
vertices in $M$, a rough estimation of $|V(G')|$ would be $O(k^5)$.  A
refined analysis would bring it to $O(k^4)$.
\begin{lemma}
  The new graph $G'$ has at most $O(k^4)$ vertices.
\end{lemma}
\begin{proof}
  Since Rule~\ref{rule:star-1} is not applicable, for each $v\in M$
  there are at most $k + 5$ cliques intersecting $N(v)$.  There are at
  most $|M| \times (k + 5) = O(k^2)$ cliques adjacent to $M$.  On the
  other hand, since Rule~\ref{rule:smooth} is not applicable, at most
  $6$ consecutive cliques can be nonadjacent to $M$.  Therefore, the
  number $t$ of cliques in $\cal K$ is at most $O(k^2)$.

  We consider first the vertices that are not in $N[M]$.  In the first
  category, we choose from each clique at most $2 k + 2$ vertices that
  are nonadjacent to $M$.  In the third and fourth categories, we
  choose from each clique at most $|M|\times (k + 1) * 4 = O(k^2)$
  vertices that are nonadjacent to $M$.  Therefore, $|V(G')\setminus
  N[M]| = O(k^2) * O(k^2) = O(k^4)$.

  Consider then vertices in $N(M)$.
  Since Rule~\ref{rule:star-2} is not applicable, for each $v\in M$
  there can be at most four cliques containing $k + 1$ or more
  vertices from $N(v)$.  There are thus at most $|M| \times 4 \le 24
  k$ such cliques.  From each of them we picked $O{|M| \choose 2} *
  O(k) + O(|M|) * O(k) * O(k) = O(k^3)$ vertices, and hence the total
  number of vertices picked from these cliques is $O(k^4)$.  Each of
  the other cliques contains at most $k$ neighbor of each vertex $v\in
  M$, and no more than $|M| \times k = O(k^2)$ vertices from $N(M)$.
  Therefore from these cliques we picked at most $O(k^2) \times O(k^2)
  = O(k^4)$ vertices that are neighbors of $M$.

  In summary, 
  \[
  |V(G')| = |V(G')\setminus N[M]| + |V(G')\cap N(M)| + |M| = O(k^4) +
  O(k^4) + O(k) = O(k^4).
  \]
  The proof is now complete.
\end{proof}

To conclude Theorem~\ref{thm:main}, it remains to verify the
equivalence between the new instance ($G', k$) and the original
instance.
Similar as the proof of Proposition~\ref{lem:contraction}, the proof
of our main lemma would be manipulating intervals.  We also take
liberty to produce a proper interval model instead of a unit interval
model: One can always turn it into a unit interval model by, say,
calling the algorithm of Bogart and West~\cite{bogart-99-proper-unit}.
Another trick we want to play is the following.  Since the set of
endpoints is finite, for any point $\rho$ in an interval model, we can
find a small positive value $\epsilon$ such that there is no endpoint
in $[\rho-\epsilon, \rho)\cup (\rho, \rho+\epsilon]$,---in other
words, there is an endpoint in $[\rho-\epsilon, \rho+\epsilon]$ if and
only if $\rho$ itself is an endpoint.  Note that the value of
$\epsilon$ should be understood as a function, depending on the
interval model as well as the point $\rho$, instead of a constant.

\begin{figure*}[h]
  \centering\footnotesize

    \includegraphics{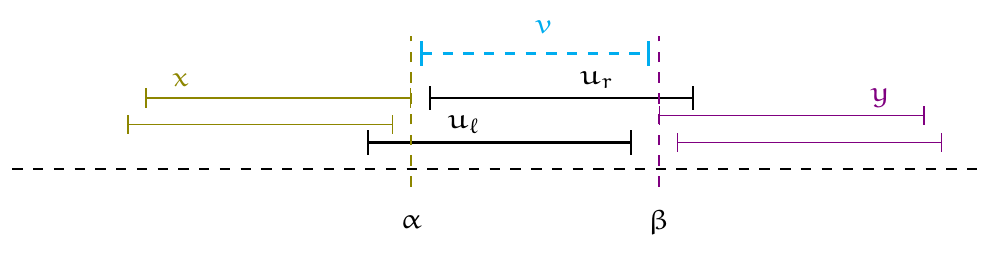}   

    (a) No vertex is adjacent to both $u_\ell$ and $u_r$ but not $v$.\\[2mm]

    \includegraphics{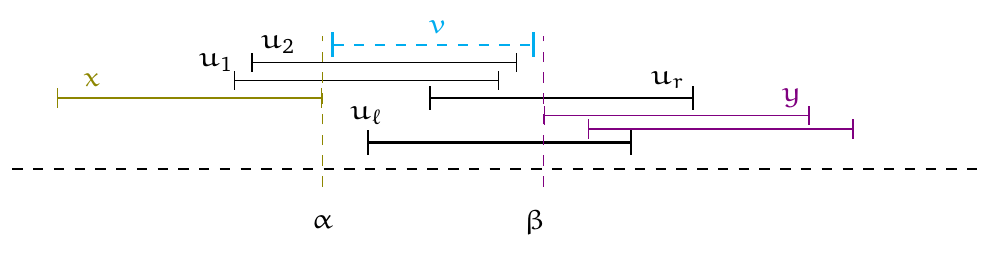}   

    (b) Some non-neighbor of $v$ is adjacent to both $u_\ell$ and
    $u_r$ from the right. 
    \caption{Illustration for Lemma~\ref{lem:fill}.}
  \label{fig:fill}
\end{figure*}

\begin{lemma}\label{lem:fill}
  If there is $V_-$ with $|V_-| \le k$ such that $G' - V_-$ is a unit
  interval graph, then $G - V_-$ is also a unit interval graph.
\end{lemma}
\begin{proof}
  We build a proper interval model for $G - V_-$ by inserting
  intervals for $V(G) \setminus V(G')$ into a unit interval model for
  $G' - V_-$ as follows.  Note that $G'$ contains $M$, and hence all
  vertices in $V(G) \setminus V(G')$ appear in $\sigma$, which is a
  proper interval ordering for $G - M$.  Let $G_0 = G'$, and let
  $G_i$, for $1\le i\le |V(G) \setminus V(G')|$, denote the subgraph
  induced by $V(G')$ and the first $i$ vertices of $V(G) \setminus
  V(G')$ in $\sigma$.\footnote{Our construction in the proof does not
    rely on any particular ordering, and it can be arbitrary.}  Since
  $G_0$ is a unit interval graph, by inductive reasoning, it suffices
  to show how to build a proper interval model for $G_i - V_-$ out of
  $G_{i - 1} - V_-$.

  Let $v$ be the $i$th vertex of $V(G)\setminus V(G')$, and let $\cal
  I$ be a unit interval model for $G_{i - 1} - V_-$.  Let $K_b$ be the
  clique in the clique partition of $G - M$ that contains $v$.  Since
  $v$ itself is not in $G'$, we have picked from $K_b$ the first $k +
  1$ and the last $k + 1$ vertices; from each of them at least one
  vertex is not in $V_-$.  There are thus vertices $u_\ell, u_r$ in
  $G_{i - 1} - (V_-\cup M)$ such that $u_\ell <_\sigma v<_\sigma u_r$.
  Assume without loss of generality $I(u_\ell)$ is to the left of
  $I(u_r)$.

  Consider first that $N_{G_i}(u_\ell)\setminus N_{G_i}[v]$ and
  $N_{G_i}(u_r)\setminus N_{G_i}[v]$ are disjoint, i.e., for every
  vertex $u\in V(G_{i - 1})\setminus V_-$ that is nonadjacent to $v$,
  the interval $I(u)$ intersects at most one of $I(u_\ell)$ and
  $I(u_r)$.  Let $x$ be the vertex in $N_{G_i}(u_\ell)\setminus
  N_{G_i}[v]$ with the rightmost interval, and let $y$ be the vertex
  in $N_{G_i}(u_r)\setminus N_{G_i}[v]$ with the leftmost interval;
  denote by $\alpha = \rp{x}$ and $\beta= \lp{y}$.  See
  Figure~\ref{fig:fill}(a).  Then
  \[
  \lp{u_\ell} < \alpha < \lp{u_r} < \rp{u_\ell} < \beta <\rp{u_r},
  \]
  and every vertex in $G_{i - 1}\setminus V_-$ with its interval
  properly contained in $[\alpha, \beta]$ has the same closed
  neighborhood as $v$ in $G_i\setminus V_-$.

  We now argue that no interval can contain $[\alpha, \beta]$.
  Suppose for contradiction $[\alpha, \beta]\subseteq I(u)$, then
  $\{u, x, y, v\}$ is a claw of $G_i$.  At least one of these four
  vertices is in $M$, because $G - M$ is a unit interval graph.
  Noting that $v$ is not in $M$, we consider which of $u, x, y$ are in
  $M$.  Note that the other vertices in $V(G_i)\setminus (V_-\cup M)$
  may or may not be in $G'$.
  \begin{itemize}
  \item Case 1, $x,y,u \in M$.  In $V_0$ there are at least $k + 1$
    vertices each of which makes a claw with $\{u, x, y\}$.  At least
    one of them is not in $V_-$ and hence $G' - V_-$ contains a claw.
  \item Case 2, $x,y \in M$ but $u\not\in M$.  Let $u\in K_a$.  We may
    assume $u<_\sigma v$; then $a = b$ or $b - 1$: Since $u$ and $v$
    are adjacent, they are either in the same clique or in two
    consecutive cliques in the partition $\cal K$.  We take $u'$ to be
    the last vertex of $K_a^1(x, y)\setminus V_-$ (it is nonempty
    because $|K_a^1(x, y)| > k$ when $u$ is not in it), and take $v'$
    to be the first of $K_b^1(\overline x, \overline y)\setminus V_-$.
    We claim that $\{u', v', x, y\}$ is always a claw in $G' - V_-$.
    By the selection, it suffices to verify that $u' v'\in E(G)$.  It
    is trivial when $a = b$, and it follows from $u\le_\sigma u'
    <_\sigma v' <_\sigma v$ and
    Proposition~\ref{lem:proper-interval-ordering} when $a = b - 1$.
  \item Case 3, $u$ and one of $x,y$ is in $M$.  We consider $x$ and
    the other is symmetric.  Let $y\in K_c$; note that $c \ne b$
    because $y$ and $v$ are nonadjacent.  We may assume $b < c$.
    We take $v'$ to be the first of $K_b^1(u, \overline x)\setminus
    V_-$, and take $y'$ to be the last of $K_c^1(u, \overline
    x)\setminus V_-$.  They are nonadjacent because $v'<_\sigma
    v<_\sigma y \le_\sigma y'$ and $v y\not\in E(G)$.  Then $\{u, v',
    x, y'\}$ is a claw in $G' - V_-$.
  \item Case 4, only $u$ is in $M$.  Let $x\in K_a$ and $y\in K_c$;
    note that $a, b$, and $c$ are all distinct because $v$, $x$, and
    $y$ are pairwise nonadjacent.  We may assume $a < b < c$ (i.e.,
    $x <_\sigma v <_\sigma y$).
    We take $x'$ to be the first of $K_{a}^1(u)\setminus V_-$, and
    take $y'$ to be the last of $K_{c}^1(u)\setminus V_-$.  They are
    clearly adjacent to $u$ but nonadjacent to each other.
    \begin{itemize}
    \item[4.1.]  If $b > a + 1$, then we take $v'$ to be the first
      of $K_{b}^1(u)\setminus V_-$; it is nonadjacent to $x'$.  It is
      also nonadjacent to $y'$ because $v'<_\sigma v<_\sigma y
      \le_\sigma y'$ and $v y\not\in E(G)$.
    \item[4.2.] Otherwise, we take $v'$ to be the first of $K_{b}^3(u,
      x')\setminus V_-$.  Note that $x'\le_\sigma x <_\sigma v$; by
      Proposition~\ref{lem:proper-interval-ordering} $x'$ is
      nonadjacent to $v$.  As a result, $K_{b}^3(u, x')\setminus V_-$
      is nonempty and $v'<_\sigma v$.  Together with $v<_\sigma y
      \le_\sigma y'$ and $v y\not\in E(G)$, we have $v' y' \not\in
      E(G)$.  The definition of $K_{b}^3(u, x')$ implies $x' v'\not\in
      E(G)$.
    \end{itemize}
    Therefore, $\{u, v', x', y'\}$ is always a claw in $G' - V_-$.
  \item Case 5, only one of $x, y$ is in $M$.  We consider $x$ and the
    other is symmetric.  Then $v u y$ is a $P_3$ in $G_i - V_-$; we
    may assume $v<_\sigma u<_\sigma y$.  Clearly, $v$ and $y$ are
    nonadjacent and hence in different cliques.  
    \begin{itemize}
    \item[5.1.] If none of them is in the same clique as $u$, then
      $u\in K_{b + 1}$ and $y\in K_{b + 2}$.  We take $v'$ to be the
      last of $K_{b}^1(\overline x)\setminus V_-$; take $u'$ to be the
      last of $K_{b + 1}^2(x, v')\setminus V_-$; and take $y'$ to be
      the first of $K_{b + 2}^1(\overline x)\setminus V_-$.  The
      vertex $v'$ is clearly adjacent to $u'$ but not $y'$.  To see
      that $v'$ is not adjacent to $y'$, note $u \le_\sigma u'
      <_\sigma y' \le_\sigma y$.
    \item[5.2.] Otherwise, $y\in K_{b + 1}$ and $u$ is in either
      $K_{b}$ or $K_{b + 1}$.  Assume without loss of generality that
      $u\in K_{b + 1}$.  We take $u'$ to be the first of $K_{b +
        1}^1(x)\setminus V_-$, take $v'$ to be the first of
      $K_{b}^4(x, u')\setminus V_-$, and take $y'$ to be the last of
      $K_{b + 1}^1(\overline x)\setminus V_-$.
    \end{itemize}
    In either case, $\{u', v', x, y'\}$ is a claw of $G' - V_-$.
   \end{itemize}
   Therefore, if no interval of $\cal I$ is contained in $[\alpha,
   \beta]$, then making $I(v) = [\alpha + \epsilon, \beta - \epsilon]$
   would make a proper interval model for $G_{i}$.  Otherwise let
   $[\alpha', \beta']$ be such an interval contained in $[\alpha,
   \beta]$; we can make $I(v) = [\alpha' + \epsilon, \beta' +
   \epsilon]$.

  In the rest there exists at least one non-neighbor $u$ of $v$ such
  that $I(u)$ intersects both $I(u_\ell)$ and $I(u_r)$.  We argue that
  $I(u)$ cannot be contained in $[\lp{u_\ell}, \rp{u_r}]$.  Suppose
  such a vertex $u$ exists, then it is must be from $M$: because
  $u_\ell<_\sigma v<_\sigma u_r$, no vertex in the unit interval graph
  $G - M$ can be adjacent to both $u_\ell$ and $u_r$ but not $v$.
  Since the model is proper, if $I(u)\subseteq [\lp{u_\ell},
  \rp{u_r}]$, then $[\lp{u_r}, \rp{u_\ell}]\subseteq I(u)$.  But we
  had also chosen from $K$ the first $k + 1$ and the last $k + 1$
  non-neighbors of $u$.  At least one of these vertices remains in
  $G_i$ and its interval has to be intersect both $I(u_\ell)$ and
  $I(u_r)$ but not $I(u)$.  This is impossible.

  Therefore, $I(u)$ approaches $I(u_\ell)$ and $I(u_r)$ either from
  the left or the right.  We may assume without loss of generality it
  is to the right of $I(u_r)$ (i.e., $\rp{u_r}\in I(u)$); the other
  case follows by symmetry.  Let us take the non-neighbor $y$ of $v$
  in $G_i$ that has the leftmost interval containing $\rp{u_r}$; let
  $\beta =\lp{y}$.  Then $\beta = \min_u \{\lp{u} : u\in
  V(G_i)\setminus N(v),\; {u_r}\in I(u)\}$, and by assumption, $\beta
  \in I(u_\ell)$.  Since $u_\ell <_\sigma v <_\sigma u_r$, the vertex
  $y$ has to be from $M$.  But then we would have also chosen from $K$
  the first $k + 1$ and the last $k + 1$ non-neighbors of $y$.  At
  least one from either set is in $G_i$; let them be $u_1$ and $u_2$
  respectively.  The intervals $I(u_1)$ and $I(u_2)$ have to approach
  $I(u_\ell)$ from the left.  Again, there cannot be vertices from
  $G_i - M$ adjacent to both $u_1, u_2$ but not $v$.  We argue that
  for any vertex $x\not\in N(v)$ adjacent to $u_1$ and/or $u_2$, the
  interval $I(x)$ is disjoint from and to the left of $I(u_\ell)$
  (i.e., $\rp{x} < \lp{u_\ell}$).  We have also chosen from $K$ the
  first $k + 1$ and the last $k + 1$ vertices that are adjacent to
  neither $x$ nor $y$.  At least one of them is in $G_i - V_-$ and its
  interval has to be accommodated between ($\rp{x}, \lp{y}$).  It
  would then be properly contained in $I(u_\ell)$ if $\rp{x} >
  \lp{u_\ell}$.  Let $\alpha = \max_u \{\rp{u}: u\in V(G_i)\setminus
  N(v),\; \lp{u_1}\in I(u)\}$.  See Figure~\ref{fig:fill}(b).  The
  rest of the construction is the same as the first
  one.  
\end{proof}

\section{Implementation issues and concluding remarks}\label{sec:remarks}

In principle, each application of
Rules~\ref{rule:star-1}--\ref{rule:smooth} should be followed by a
re-calculation of the modulator: After the application, the rest of
$M$ (it loses one vertex with Rules~\ref{rule:star-1}
and~\ref{rule:star-2} but remains intact with Rules~\ref{rule:smooth})
may not be a $6$-approximation for the remaining graph.  This would
imply that it takes $O(n \cdot n m)$ time to exhaustively apply
Rules~\ref{rule:star-1}--\ref{rule:smooth}.

We have been using the approximation algorithm
\cite{cao-15-unit-interval-editing} as a black box for furnishing the
modulator.  To have a better analysis, we may have to unwrap the black
box and see a bit of how it works.  It consists of two phases.  The
first phase keeps looking for a claw, net, tent, $C_4$, or $C_5$, and
deletes all its vertices if one is found.  When none of these small
\fis{s} can be found, the algorithm enters the second phase, which
then finds an optimal solution in linear time.  The ratio is $6$
because unit interval graphs are hereditary and any optimal solution
needs to contain at least one vertex from any induced claw, net, tent,
$C_4$, or $C_5$.  Recall that whether a graph contains a claw, net,
tent, $C_4$, or $C_5$ can be decided in linear time, and if yes, one
can be detected in the same time.

Consider first Rules~\ref{rule:star-1} and~\ref{rule:star-2}, each of
which deletes a vertex from $M$.  If the deleted vertex $v$ had been
added to $M$ in the second phase of the approximation algorithm, then
the set $M - \{v\}$ is still a $6$-approximation for $G - \{v\}$, and
we need to do nothing.  Otherwise, we need to (re-)calculate a new
approximation solution for $G - \{v\}$.  Fortunately, we do not need
to start from scratch.  Recall that $v$ had been put into $M$ because
it is in some induced claw, net, tent, $C_4$, or $C_5$ found in phase
1; let $X$ be the at most six vertices of this \fis{}.  Let $M'$
denote the subset of vertices of $M\setminus X$ that are added in the
first phase; they are still \emph{good} in the sense that they still
form vertex-disjoint claws, nets, tents, $C_4$'s, and $C_5$'s.
Therefore, we may start the approximation algorithm with $M'$ as the
partial solution.  Note that every claw, net, tent, $C_4$, or $C_5$ in
$G - \{v\} - M'$ needs to intersect $X\setminus\{v\}$.  Therefore, we
can find at most six vertex-disjoint claws, nets, tents, $C_4$'s, and
$C_5$'s, which can be done in $O(m) * 6 = O(m)$ time.  We put all the
vertices in the found subgraphs, and redo the second phase in another
$O(m)$ time.  Consequently, we can produce a $6$-approximation for the
new graph $G - \{v\}$ in $O(m)$ time.

The situation for Rule~\ref{rule:smooth} is actually simpler.  It does
not touch $M$, and thus all vertices added in the first phase remain
good.  We can redo the second phase in $O(m)$ time to produce an
approximation solution for the new graph.\footnote{This step is not
  really necessary, because it can be proved that $M$ remains a
  $6$-approximation of the new graph after the application of
  Rule~\ref{rule:smooth}.  Moreover, the impact of
  Rule~\ref{rule:smooth} on the interval model and the clique
  partition of $G - M$ is local, and a new model and a new partition
  can be easily recovered.  For the simplicity of presentation, the
  form of Rule~\ref{rule:smooth} given in Section~\ref{sec:reductions}
  deletes only one clique.  One can show that it can be easily adapted
  to contracting all but $6$ cliques in a sequence of cliques in $\cal
  K$ that are nonadjacent to $M$.  Moreover, all the cliques can be
  handled in one run, in linear time.  To prove these facts, however,
  we need to revisit the approximation algorithm
  \cite{cao-15-unit-interval-editing} with all the details, which we
  omit to not blur the focus of the current paper.}

Therefore, we can apply each reduction rule and presently recover the
modulator in $O(m)$ time.  On the other hand, since each application
of a reduction rule deletes at least one vertex from the graph, they
can be applied at most $n$ times.  The total running time of them is
$O(n m)$.  The picking of the vertices can be easily done in $O(k^3
m)$ time.  Note that if $n < k^4$, then we do not need to do anything
at all.  Therefore, the running time of the whole kernelization
algorithm is $O(n m)$.  

The primary concern of a kernelization algorithm is surely the kernel
size.  Kernelization algorithms may not completely solve the instance,
and then they are followed by other algorithmic approaches.  Its
applicability would thus be limited if the running time is too high.
In literature, however, very little attention has been paid to the
running time of most kernelization algorithms, and most of the time, a
detailed analysis is omitted.  (Most of them are trivially
polynomial.)  Once we aim for ``efficient'' kernelization with lower
polynomial running time, we need to reconsider the tools we can use.
For example, for all vertex deletion problems to hereditary graph
class, we have a trivially correct reduction rule that deletes all
vertices not participating in any forbidden induced subgraphs.  This,
however, is usually very time-consuming (it takes normally $n^{|X|}$
where $X$ is the largest forbidden induced subgraph) and thus should
be avoided.

As a final remark, properties of the approximation algorithm
\cite{cao-15-unit-interval-editing} may be further exploited to
sharpen the analysis of the kernel size of our kernelization
algorithm.  But it would very unlikely lead to one with $o(k^2)$
vertices.  We leave the existence of a linear-vertex kernel for the
\uivd{} problem as an open question.

\paragraph{Acknowledgments.} The authors want to thank Jinshan Gu,
R.~B.~Sandeep, and Jie You for fruitful discussions during an early
stage of this project.

{
  \small 
  \bibliographystyle{plainurl}
  \bibliography{../../journal,../../main}
}
\end{document}